\documentclass[PP]{packages/AEA}

\usepackage{packages/extra_packages}

\title{SmartDCA superiority}
\author{Calvet, Emmanuel$^\circ$\thanks{emmanuel.calvet@usherbrooke.ca, Universit\'e de Sherbrooke}, \, Herranz-Celotti, Luca$^\circ$\thanks{luca.celotti@usherbrooke.ca, Universit\'e de Sherbrooke} \,  and Valimamode, Karim\thanks{karim\_valimamode@hotmail.com}}
\date{\today}
\pubMonth{08}
\pubYear{2023}
\pubVolume{1}
\pubIssue{53}

\begin{document}

\begin{abstract}
Dollar-Cost Averaging (DCA) is a widely used technique to mitigate volatility in long-term investments of appreciating assets. However, the inefficiency of DCA arises from fixing the investment amount regardless of market conditions. In this paper, we present a more efficient approach that we name SmartDCA, which consists in adjusting asset purchases based on price levels. The simplicity of \mbox{SmartDCA} allows for rigorous mathematical analysis, enabling us to establish its superiority through the application of Cauchy-Schwartz inequality and Lehmer means. We further extend our analysis to what we refer to as \mbox{$\rho$-SmartDCA}, where the invested amount is raised to the \mbox{power of $\rho$}. We demonstrate that higher values of $\rho$ lead to enhanced performance. However, this approach may result in unbounded investments. To address this concern, we introduce a bounded version of SmartDCA, 
taking advantage of two novel mean definitions that we name \mbox{quasi-Lehmer means}. 
The bounded SmartDCA is specifically designed to retain  its superiority to DCA. 
To support our claims, we provide rigorous mathematical proofs and conduct numerical analyses across various scenarios. The performance gain of different \mbox{SmartDCA} alternatives is compared against DCA using data from S\&P500 and \mbox{Bitcoin}. The results consistently demonstrate that all SmartDCA variations yield higher long-term investment returns compared to DCA. 
\end{abstract}

\maketitle

\def\thefootnote{$\circ$ \, }\footnotetext{These authors contributed equally to this work}

\section{Introduction}

Dollar Cost Averaging (DCA) is a common investment strategy, where an investor puts regularly a constant fraction of her wealth into the same asset to outperform the return that she would get by putting all her capital in an asset at once \cite{KIRKBY20201168, smith2018another}. Other more sophisticated investment strategies have been shown experimentally to outperform the DCA \cite{dunham2012building, kapalczynski2021effectiveness, payne2013l}. However, it is often hard to prove that these strategies can systematically outperform the DCA.

In this work, we provide mathematical proof for a simple investing strategy that can outperform the DCA in any market condition, which we call the SmartDCA. Essentially it consists in
regularly investing an amount of money that is inversely proportional to the current price of the asset. 
Similar approaches have been used before, without a rigorous mathematical justification \cite{dunham2012building, kapalczynski2021effectiveness, payne2013l}, and thus only claiming empirically their superiority, in simulation or on historical data. 
Instead, we show that it is possible to have investment strategies that are provably better than DCA, without any assumption on the market. For that, we had to introduce new definitions of means that are generalizations of the Lehmer mean, and for that reason, we call quasi-Lehmer means, in analogy with quasi-arithmetic means \cite{bullen2003handbook, bullen2013means, matkowski2020beckenbach}. We show on historical S\&P500 and Bitcoin data, that investing through the \mbox{SmartDCA} systematically improves the return on investment with respect to DCA.
\section{SmartDCA: using the price ratio} \label{dca}

We start by proving mathematically how regularly investing in an asset a quantity of money that is inversely proportional to the current price, results in a better cost per unit of the asset. For the sake of clarity, we first show how it is the case when only two trades are considered in Sec.~\ref{sec:2trades}, and then we prove it for any number of investment events in Sec.~\ref{sec:mbuys}.

\subsection{Mean cost of 2-buying times}
\label{sec:2trades}

We study three scenarios of investment and compare them mathematically in order to prove what is the best strategy for investing in an arbitrary asset. Consider buying a good, for example gas, every month, anything that one would recurrently buy. The question we ask is whether it is more advantageous to buy a fixed quantity of gas (Regular Investing), than buying varying amounts of gas at a fixed cost (DCA). Next, we explore the intuition that when the price of gas is low, it is in our interest to buy more, and less when the price is high (SmartDCA). 
 
\subsubsection{First scenario: Regular Investing (RI)}
This scenario consists in buying at two times, $t_1$ and $t_2$, a quantity $q$ of an asset. At $t_1$ the price is $p_1$, and $p_2$ is the price at $t_2$. The total cost $c_{tot}$ that is spent is:
\begin{equation}
    c_{tot}=p_1q + p_2q   
\end{equation}
For example, let's say one bought half a litre of gas $q=0.5L$ at each time step, $p_1=0.5\$$, $p_2=1.5\$$ the total quantity is $q_{tot}=2q=1L$, for a total cost of $c_{tot}=2\$$. For the first scenario, the average cost per litre of gas $\mu_{RI}=c_{tot}/q_{tot}$ is:
\begin{align}
    \mu_{RI} = \frac{p_1q + p_2q}{2q} \\
    \quad = \frac{p_1 + p_2}{2} \label{eq:arithmetic}
\end{align}
In that case, the regular investing strategy turns out to give an average cost Eq.~(\ref{eq:arithmetic}) equal to the \textit{arithmetic mean} of the prices. With the gas example, the average price would be $\mu_{RI}=1\$/L$. 

\subsubsection{Second scenario: Dollar Cost Average (DCA)}
Now let's say that instead of buying a fixed amount of gas, one decided to always spend the same amount of money, at a fixed cost $c$: this is the Dollar Cost Average (DCA). At each time step, the quantity $q$ that is bought, is different. This quantity depends on the price $p$ of one unit of asset:
\begin{equation}
    q=c/p
\end{equation}
In the gas example, if the price of one litre is $2\$/L$, and one decides to buy for a cost $c=1\$$, then trivially, he bought a quantity $q=0.5L$. This time the total cost is:
\begin{equation}
  c_{tot}=c + c=2c  
\end{equation}
The total quanity $q_{tot}$ of the asset is:
\begin{equation}
    q_{tot}=c/p_1 + c/p_2
\end{equation}
So, in the end, for the second scenario, the price per unit of asset is simply:
\begin{align}
    \mu_{DCA}= \frac{2c}{c/p_1 + c/p_2} = \frac{2p_1p_2}{p_1 + p_2} \label{eq:harmonic}
\end{align} 
The average price in the case of the DCA is well known in statistics, as it is the \textit{harmonic mean}. This is interesting because Eq.~(\ref{eq:arithmetic}) and Eq.~(\ref{eq:harmonic}) are related in a well-known inequality, explaining why DCA is superior to RI:
\begin{align}
    \frac{p_1 + p_2}{2} \ge \frac{2p_1p_2}{p_1 + p_2}
\end{align}
\noindent which means that we pay less for the same amount of asset.

Following our example with gas, the average price for the DCA is $\mu_{DCA} = 0.75\$/L$, which is indeed inferior to $ \mu_{RI}=1\$/L$.

Note that when $p_1=p_2$ the two scenarios give the same price per unit, by replacing in Eq.~(\ref{eq:arithmetic}) and Eq.~(\ref{eq:harmonic}):
\begin{align}
    \mu_{RI} =& \frac{2p_1}{2}= p_1 \\
    \mu_{DCA} =& \frac{2{p_1}^2}{2p_1}=p_1
\end{align}
\subsubsection{Third scenario: SmartDCA}
Now we are going to explore a last scenario, where instead of buying at a fixed cost, irrespective of market conditions, we are going to optimize things by applying the following logic: if  the price is higher, then we want to buy less, and vice-versa.
We call this method the \mbox{SmartDCA}. First, at $t_1$, we buy for a base cost $c_1=c_b$. Next, at $t_2$ we buy for a cost that will depend on the price movement and the base price:
\begin{equation}
  c_2=c_b\frac{p_1}{p_2}  
\end{equation}
Now the total quantity $q_{tot}$ of the asset to buy is:
\begin{equation}
    q_{tot}=\frac{c_b}{p_1} + \frac{c_b}{p_2}\frac{p_1}{p_2}
\end{equation}
Which, with some arithmetic, gives:
\begin{equation}
    q_{tot}=c_b \left( \frac{{p_1}^2 + {p_2}^2}{p_1{p_2}^2} \right)    
\end{equation}
The total cost of these transactions is:
\begin{align}
    c_{tot}=c_b + c_b\frac{p_1}{p_2}    
\end{align}
\noindent  that results in the price per unit $\mu_{SmartDCA}=c_{tot}/q_{tot}$:
\begin{equation}
    \mu_{SmartDCA}=\frac{c_b + c_b\frac{p_1}{p_2}}{c_b \left( \frac{{p_1}^2 + {p_2}^2}{p_1{p_2}^2}\right)}
\end{equation}
We rewrite to obtain the final form:
\begin{equation}
    \mu_{SmartDCA} = p_1p_2\frac{p_1+p_2}{{p_1}^2 + {p_2}^2} \label{eq:adaptive}
\end{equation}
This mean is inversely proportional to the \textit{contraharmonic mean}.
With our previous gas prices, we have a price per unit of $\mu_{SmartDCA} =0.6\$/L$, which is inferior to $\mu_{DCA}$.
Note that when $p_1=p_2$ all three scenarios still give the same price per unit, using Eq.~(\ref{eq:adaptive}):
\begin{align}
    \mu_{SmartDCA} = \frac{{p_1}^3+{p_2}^3}{{p_1}^2 + {p_1}^2}\\
    = \frac{2{p_1}^3}{2{p_1}^2}= p_1 \\
    \mu_{SmartDCA} = \mu_{DCA} = \mu_{RI} = p_1
\end{align}

\subsection{Suppremacy of SmartDCA  for \mbox{2-buys}}

To actually provide a real proof, we need to show that the inverse of the \textit{contraharmonic mean} (SmartDCA), is inferior or equal to the \textit{harmonic mean} (DCA). For ease of notation, we will consider $x=p_1$ and $y=p_2$ in the following proof. So we need to solve the following inequality:
\begin{equation}
    2\frac{xy}{x+y} \ge \frac{x^2y+y^2x}{x^2+y^2}    
\end{equation}
\noindent and with some arithmetics we obtain:
\begin{align}
    \frac{2}{x+y}\cancel{xy} \ge & \ \cancel{xy}\frac{x+y}{x^2+y^2}\\    
    2\frac{x^2+y^2}{(x+y)^2} -1 \ge & \ 0\\
    2(x^2+y^2)-(x+y)^2 \ge & \ 0
\end{align}
Simplifying:
\begin{align}
    x^2 + y^2 - 2xy \ge & \  0   \\
    (x-y)^2 \ge & \ 0
\end{align}
\noindent which is a well-known polynomial identity, and trivially, a square cannot be negative, so this is true for all $x$ and $y$. It means that the difference between the DCA and the SmartDCA scales as the square of the difference between $p_1$ and $p_2$, and as seen previously, when $p_1=p_2$, they are equal. 

\subsection{Supremacy for m-buys}
\label{sec:mbuys}

In the most general form, we want to see if after $m$ buys performed by the investor using the SmartDCA, the cost per unit of asset is better than using  DCA. Essentially the quantity to invest according to the SmartDCA has to be inversely proportional to the price, but to make it unitless we will multiply it by a reference price of our choice $p_r$, and therefore the vanilla SmartDCA suggests investing $p_r/p_i$ \mbox{at time $i$}. We prove in App.~\ref{app:smd} that:

\begin{restatable}[SmartDCA superiority over DCA]{theorem}{sdca}
\label{thm:sdca} Over $m$-buying events, investing through the SmartDCA results in better price per unit than investing through DCA.
\end{restatable}


\subsection{Generalization to $\rho$-SmartDCA}

If we want to be even more general, let's consider investing $(p_r/p_i)^\rho$ regularly at the $i$-th buying event, and let's call the resulting strategy the $\rho$-SmartDCA, which gives an average price per unit of the asset $\mu_\rho$. Again, $p_r$ is the price of reference and will be kept constant. The interest in using such exponent is that when the price is above the price of reference, it will result in even less investing, and when the price is inferior, it will exponentially increase. In the following, we will demonstrate that this strategy gives superior results. After the Theorem statement, we show the proof of superiority. We use it to introduce the concept of \textit{Lehmer mean} \cite{bullen2003handbook}, necessary for the even more general Thm.~\ref{thm:bsdca} that will follow.

\begin{restatable}[$\rho$-SmartDCA improves with higher $\rho$]{theorem}{usdca}
\label{thm:usdca}  Investing through the $\rho$-SmartDCA, higher $\rho$ results in better price per unit, over $m$ buying events.
\end{restatable}

\begin{proof}
We proceed as before, at each time step, we invest an amount proportional to a base cost $c_b$,  take the ratio of the reference price $p_r$ and the current price, and then raise it to the power of $\rho$. We therefore buy a total quantity $q$ for a total price $c$:
\begin{align}
    c &=c_b\left(\frac{p_r}{p_1}\right)^{\rho} + c_b\left(\frac{p_r}{p_2}\right)^{\rho}   + \cdots \notag \\ 
    & \quad \cdots + c_b\left(\frac{p_r}{p_m}\right)^{\rho} \\
    q &=\frac{c_b}{p_1}\left(\frac{p_r}{p_1}\right)^{\rho} + \frac{c_b}{p_2}\left(\frac{p_r}{p_2}\right)^{\rho}   + \cdots \notag \\ 
    & \quad \cdots + \frac{c_b}{p_m}\left(\frac{p_r}{p_m}\right)^{\rho}
\end{align}
Since $p_r$ is constant, we will use the ratio $r_i=p_r/p_i$ as our base of reference for the calculus:
\begin{align}
    q=\frac{c_b}{p_r}\Big(\sum_{i=1}^m{r_i^{\rho+1}}\Big)\\
    c =c_b\Big(\sum_{i=1}^m{r_i^{\rho}}\Big)
\end{align}
\noindent and we are interested in the mean price per unit of asset:
\begin{align}
     \mu_\rho=& \ \frac{c}{q} =  \ p_r\frac{\sum_{i=1}^m {r_i^{\rho}} }{\sum_{i=1}^m{r_i^{\rho+1}}}
\end{align}
Now, notice the similarity with the Lehmer mean \cite{bullen2003handbook}:
\begin{align}
    L_\rho(\boldsymbol{x})= \frac{\sum_{i=1}^m x_i^{\rho} }{\sum_{i=1}^m x_i^{\rho-1}}
\end{align}
We will make use of the fact that $\rho\leq\rho'\implies L_\rho(x)\leq L_{\rho'}(x)$ \cite{wikilehmer}.
Using  the notation $\boldsymbol{r} = p_r(\boldsymbol{1/p}) = p_r(1/p_1, 1/p_2, \cdots, 1/p_m)$,
 If  $\rho\leq\rho'$ we can write the mean cost of a unit of the asset as:
\begin{align}
    \mu_\rho
    =& \ p_r\frac{\sum_{i=1}^mr_i^{\rho} }{\sum_{i=1}^m r_i^{\rho+1}}\\
    =& \ \frac{p_r}{L_{\rho+1}(\boldsymbol{r})}\\
    \geq& \ \frac{p_r}{L_{\rho'+1}(\boldsymbol{r})}=\mu_{\rho'}
\end{align}
\noindent and therefore, the higher the $\rho$, the better the average price per unit. QED.

\end{proof}

Notice that $\rho=-1$ corresponds to the Regular Investing strategy and $\rho=0$ to the DCA. Therefore, both can be considered to belong to a more general family of investing strategies, the $\rho$-SmartDCA. As a consequence of Theorem~\ref{thm:usdca}, any $\rho>1$ will outperform the SmartDCA, any $\rho>0$ will outperform the DCA and  any $\rho>-1$ will outperform Regular Investing.  Moreover, if $\rho \rightarrow \infty$, $L_\infty(\boldsymbol{r})=\max\{\boldsymbol{r}\}$ \cite{Sykora2009}, thus $\mu_\infty=p_r/\max\{\boldsymbol{r}\}= \min\{\boldsymbol{p}\}$, and we have a lower bound for the best strategy price per unit.
Therefore, all these strategies are connected by their price per unit $\mu_\rho$, in the following inequality:
\begin{equation*}
\boxed{
\begin{array}{rcl}
\mu_{-1} \geq \mu_0 \geq \mu_{\rho\geq0} \geq  min\{\boldsymbol{p} \}
\end{array}
}
\end{equation*}






\subsection{Generalization to $(f)\rho$-SmartDCA}

Now, the drawback of using the \mbox{$\rho$-SmartDCA} as above is that it will potentially ask you to invest more money than you have if the price is sufficiently low. Indeed, we have $(p_r/p_i)^\rho\rightarrow \infty$ as $p_i \rightarrow 0$, which happens exponentially faster with $\rho>1$. For that reason, we propose two modifications that can be defined to have a maximal investment amount. In the \textit{in} version, for every buying event, we invest $c_bf((p_r/p_i)^\rho)$,  and in the second \textit{out} version we invest $c_bf(p_r/p_i)^\rho$, where $in, out$ are just a reminder of the positioning of the power with respect to $f$.  Notice that the results in this section hold for any $f$ positive monotonic increasing, and for example, if $f$ is the identity, we recover the unbounded \mbox{$\rho$-SmartDCA} as a special case. However, we are interested in \mbox{$f$ bounded}, such as the function  $tanh$, since in that case the strategy could only ask the investor for a maximal investment of $c_b$. To be able to tackle this more general case, we introduce two new means that we call \textit{quasi-Lehmer means}, taking the form:
\begin{align}
    L^{(in)}_{\rho+1}(\boldsymbol{x})=& \ \frac{\sum_{i=1}^m x_if(x_i^\rho) }{\sum_{i=1}^m f(x_i^\rho)}\\
    L^{(out)}_{\rho+1}(\boldsymbol{x})=& \ \frac{\sum_{i=1}^m x_if(x_i)^\rho }{\sum_{i=1}^m f(x_i)^\rho}
\end{align}
\noindent where the naming choice is to draw the parallel with \textit{quasi-arithmetic means}.
In fact, we prove in App.~\ref{app:newmeans} that:

\begin{restatable}[quasi-Lehmer means monotonicity]{theorem}{qlm}
\label{thm:qlm} If $\rho\leq \rho'$ and $f$ is positive and monotonic increasing then $L^{(out)}_\rho(\boldsymbol{x})\leq L^{(out)}_{\rho'}(\boldsymbol{x})$, and therefore $L^{(out)}_\rho(\boldsymbol{x})$ is monotonic increasing with $\rho$. However, $L^{(in)}_\rho(\boldsymbol{x})$ is not in general monotonic \mbox{increasing with $\rho$}.
\end{restatable}

We show in App.~\ref{app:newmeans} that a similar Theorem holds for what we call the \textit{quasi-Lehmer moments} and the \textit{quasi-Lehmer expectations}. We use this Theorem to prove in App.~\ref{app:outsuperiority}, that  the  \textit{out} version of the \mbox{$(f)\rho$-SmartDCA} improves with $\rho$:

\begin{restatable}[The higher the $\rho$, the better the $(f)\rho$-SmartDCA$^{(out)}$]{theorem}{bsdca}
\label{thm:bsdca} Investing through the  \mbox{$(f)\rho$-SmartDCA$^{(out)}$} results in better price per unit over $m$-buying events, if we increase $\rho$.
\end{restatable}

Given that we recover the DCA strategy as we set \mbox{$\rho=0$}, the last Theorem also implies that \mbox{$(f)\rho$-SmartDCA$^{(out)}$} outperforms DCA. 
If $f$ is chosen to be bounded, it does so without incurring into the risk of exorbitant investments that could be suggested by the unbounded \mbox{$\rho$-SmartDCA}. 
Given that we proved above that DCA outperforms RI, it follows that 
 $(f)\rho$-SmartDCA$^{(out)}$ also outperforms RI.
\section{Numerical Analysis}

\subsection{$(f)\rho$-SmartDCA outperforms DCA for any $\rho$ experimentally}

We show in Fig.~\ref{fig:rhos} the effect that $\rho$ has on $\mu$, the price per unit  of the asset, for $\rho$-SmartDCA, \mbox{$(f)\rho$-SmartDCA$^{(in)}$} and \mbox{$(f)\rho$-SmartDCA$^{(out)}$}. We can see that all of them outperform the DCA, in the sense that all have lower price per unit $\mu$. The price is simulated as samples from a uniform distribution between zero and two. This stresses that our strategies outperform the DCA even in the lack of market trends. The unbounded \mbox{$\rho$-SmartDCA} achieves the lowest price per unit, but it results in absurd investments required when the prices are very low, as can be seen in the lower panels. \mbox{$(f)\rho$-SmartDCA$^{(out)}$} tends to achieve better $\mu$ than \mbox{$(f)\rho$-SmartDCA$^{(in)}$}, with the added advantage of being always provably better than DCA. The three columns in the plot correspond to three \mbox{different $f$}: $tanh$, $sigmoid$ and what we call the \mbox{$sin$-$1$}, a function that goes from zero to one as a $sin$, and then stays at one.

\begin{figure}
  \begin{minipage}[c]{0.48\textwidth}
  \includegraphics[width=1.\textwidth]{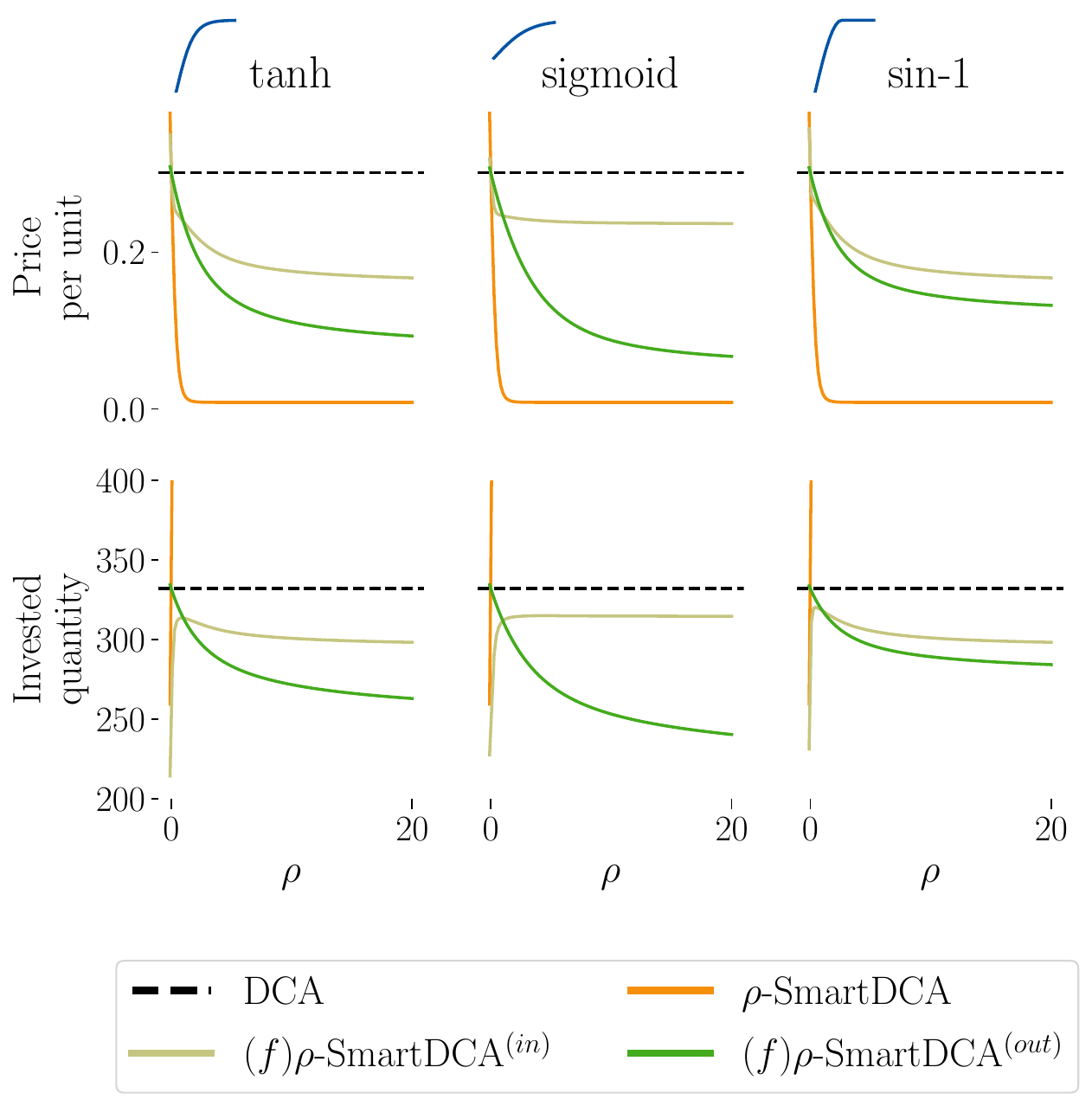}
  \end{minipage}
  \hfill
  \begin{minipage}[c]{0.48\textwidth}
    \caption{\textbf{SmartDCA outperforms DCA for any $\rho\geq0$ choice.} We simulate the behaviour of an investor that puts money regularly into an asset. The prices of the asset at buy time are 100 samples from a uniform distribution from zero to two. We plot the price per unit and invested quantity for three choices of $f$. All SmartDCA variants manage to buy at a lower price than the DCA, for any choice of $\rho$. However, unbounded SmartDCA (orange) can suggest exorbitant amounts to invest if the price is low enough, as seen in the lower panels. Only $(f)\rho$-SmartDCA$^{(out)}$ (green) is provably and unconditionally better than DCA, without investment amounts that blow up with low prices, as we prove mathematically and confirm through the plots. Even if the invested quantities appear lower for $(f)\rho$-SmartDCA$^{(out)}$ and \mbox{$(f)\rho$-SmartDCA$^{(in)}$} than for DCA, they can be matched with a higher $c_b$, where $c_b=p_r=1$ in the plot.}
    \label{fig:rhos}
  \end{minipage}
\end{figure}

\begin{figure}
    \includegraphics[width=.9\textwidth]{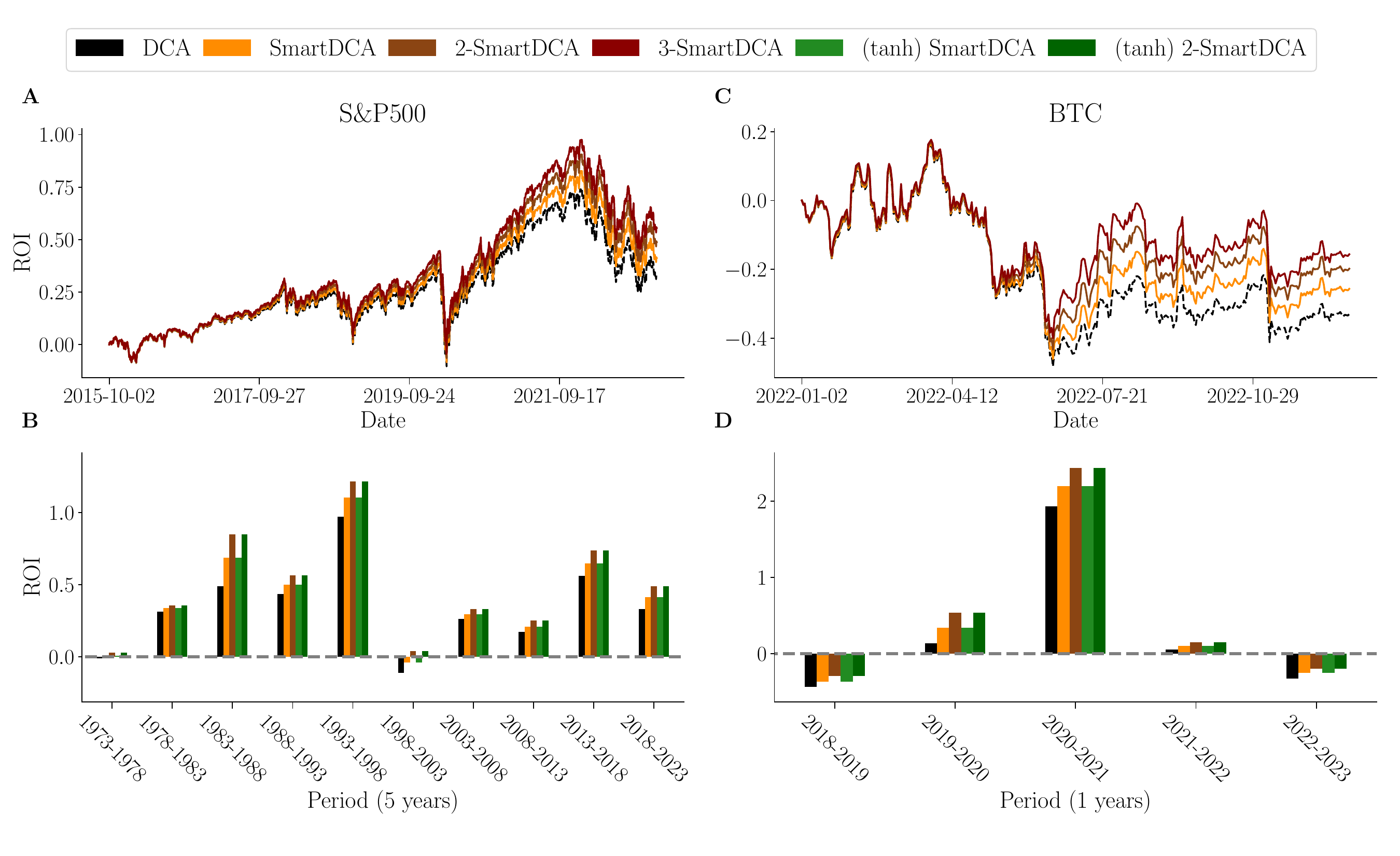}
    \caption{\textbf{ROI on S\&P500 and Bitcoin: SmartDCA outperforms DCA.} We simulate a buyer using DCA, $\rho$-SmartDCA with $\rho\in\{1,2,3\}$, and $(f)\rho$-SmartDCA$^{(out)}$ with $\rho\in\{1, 2\}$ and $f=tanh$. All SmartDCA variants achieve higher Return on Investment (ROI) than the DCA, for any $\rho$ choice, on both the S\&P500 and Bitcoin. Notice that the $(f)\rho$-SmartDCA$^{(out)}$ variants are not plotted in the upper panels because they were completely overlapping with their unbounded version. This is the case since the high prices of both studied assets  would activate the linear part of the $tanh$, resulting in very similar results as the unbounded  $\rho$-SmartDCA. Lower panels show that the same improvement over DCA can be seen in every single period of five years for S\&P500, and every single period of one year for Bitcoin.}
    \label{fig:sp500bitcoin}
\end{figure}

\begin{table*}[h!]
\begin{tabular}{lrrrr} \toprule
Strategy & $\mu \ (\$/BTC)$ & $q_{tot} \ (BTC)$ & $c_{tot} \ (\$)$ &  ROI \\ \midrule
DCA & \cellcolor{orange!25} 10942.8 & 0.166 & 1827 & \cellcolor{orange!25} 1.518 \\
(sig+) SmartDCA & \cellcolor{GoodGreenLight} 8893.6 & 0.058  \cellcolor{GoodGreenLight} & \cellcolor{GoodGreenLight} 516.3 & \cellcolor{GoodGreenLight} 1.868 \\
(tanh) SmartDCA &  7134.8 & $2.340\cdot10^{-5}$ & \cellcolor{red!25} 0.166 &  2.328 \\
3-SmartDCA & \cellcolor{GoodGreen} 4790.5 & 0.873 & \cellcolor{red!25} 4180.9 & \cellcolor{GoodGreen}3.46 
\end{tabular}
\caption{\textbf{Finetuning is required to buy substantial quantities. } If unchecked, the SmartDCA variants can end up buying only small quantities of the asset. Here we show that the DCA buys at worse price per unit than $(tanh)$ SmartDCA but the latter buys only a very small total amount of Bitcoin ($0.166\$$), from 2017 to 2023. Instead the $3$-SmartDCA, will suggest to invest an amount of capital potentially much higher than foreseen ($4180\$$), despite an excellent ROI. Adjusting the slope and the center of a sigmoid yearly with data of the maximal and minimal price of the previous year, allows to keep a better price per unit and ROI than DCA, while maintaining the total quantity bought to $516.3\$$, below our chosen maximum of $1827\$$, that corresponds to one dollar per day.}\label{tab:sigplus}
\end{table*}

\subsection{Improvements on S\&P500 and Bitcoin Investments}

In this section, we  backtest this family of strategies using real-life case scenarios. 
Since investment strategies are of interest on an overall up-trend, we test them on assets with long-term appreciating values. 
We are going to use the stock market S\&P500, an Exchange Traded Fund which measures the market capitalization of the United States 500 largest corporations. This asset class is particularly attractive for DCA investors because its estimated annualized total return is around $9\%$ (from January 1996 to June 2022) \cite{Gupta2022}. The other asset we use is Bitcoin (BTC) \cite{Nakamoto2008}, a digital crypto-currency based on a decentralized peer-to-peer electronic cash system. Bitcoin has grown in popularity over the last few years and has seen its price skyrocket, with an average annual return of around $80\%$ \cite{lazy2023}. 

For the backtest, we fix the price of reference to the first price obtained in the time series $p_r=p_1$, and we test $\rho=\{1,2,3\}$ along with the function $tanh^{(out)}$. To evaluate the performance of these strategies, we measure the Return on Investment (ROI), computed as the net gains divided by the costs. The simulations are performed with \textit{kiwano-portfolio} \cite{portfolio2022}, using the setting \textit{fast\_backtesting}. \textit{kiwano-portfolio} is an open-source trading software created by the authors. As it can be seen in Fig.~\ref{fig:sp500bitcoin}, all \mbox{$\rho$-SmartDCA} and $(f)\rho$-SmartDCA variants outperform the DCA ROI, as we expected given our mathematical proofs. We show in the upper panels how the distance with DCA compounds over time. We also show in the bottom panels that the improvement over DCA can be seen in all five-year periods considered for the S\&P500 and all one-year periods for Bitcoin. One can note that even for periods of loss, the SmartDCA strategies still manage to lose less than the DCA.

\subsection{Adapting $f$ on past data}

Note that if the shape of $f$ and the reference price $p_r$ are not chosen carefully, the final quantity bought can be very low, even if it was bought at an excellent price per unit. To address this issue, we propose adapting the sensitive part of a sigmoid curve to the maximal and minimal prices of the previous year:
\begin{equation}
f(x) = \text{sigmoid}((x-x_0)/\lambda)
\end{equation}
\noindent with $x_0=(y_{max} + y_{min}) / 2$ and $\lambda=(y_{max} - y_{min})/8$. We define $y_{max} = \max_i 1/p_i$ and $y_{min} = \min_i 1/p_i$ over the prices of the previous year. We refer to the resulting strategy as the \mbox{$(sig+)$ SmartDCA}. 
As you can see in Table~\ref{tab:sigplus}, assuming a base cost investment of $c_b=1\$$ in Bitcoin each day,
\mbox{3-SmartDCA} achieves the best ROI and $\mu$ (price per unit), but it comes with an investment over twice our base cost. On the other hand, the bounded \mbox{$(tanh)$ SmartDCA} achieves the second best ROI and  $\mu$, but buys a negligible quantity of Bitcoin. By adapting the shape of the $sigmoid$, we manage to maintain a better ROI and $\mu$ than DCA, and purchase an amount significantly closer to our desired base cost. Finally, one can observe that multiplying by three the base amount of dollars invested per day with \mbox{$(sig+)$ SmartDCA}, would roughly result in the same final quantity of asset obtained with DCA, while keeping a lower $\mu$ and higher ROI.

\section{Discussion and Conclusion}
We showed that the DCA, and Regular Investing, are elements of a broader category of strategies that we called \mbox{$(f)\rho$-SmartDCA}, with $f$ a positive monotone increasing function, and $\rho$ the exponent applied to a modulator for a reference price $p_r$, and a base cost $c_b$, such that the suggested investment at time $i$ is $c_b f(p_r/{p_i})^\rho$. For each of these strategies, we computed the average price per unit $\mu_\rho$ and were able to demonstrate mathematically that they follow a decreasing order with $\rho$:
\begin{align}
    \mu_0 \ge \mu_1 \ge \mu_2 \ge ... \ge \mu_\rho    
\end{align}
As such, we proved that the DCA corresponds to a $0$-SmartDCA, outperformed by all $(f)\rho$-SmartDCA for \mbox{$\rho>0$}. Notice 
that the buying events could be placed randomly in time and $\rho$-SmartDCA would still outperform the DCA, since a regular time assumption was never used in the Theorems. The regularity in the investments is to overtake human psychology and the tendency to go into investments when they are popular and therefore, likely to tip.

Moreover, we introduced the \textit{quasi-Lehmer means} and its generalizations to be able to prove that a wide family of \mbox{$(f)\rho$-SmartDCA} mathematically outperforms the DCA. Finally, we empirically confirmed our theoretical findings on random data and on the S\&P500 and Bitcoin historical data: \mbox{$(f)\rho$-SmartDCA} is superior to DCA.

To finish, in Appendix~\ref{app:newmeans}, we were able to further generalize our proof by the use of quasi-Gini means for Thm.~\ref{thm:qlmm} and  Thm.~\ref{thm:qlem}, and in future work one could potentially design even more universal investment strategies based on these new theorems. 

\bibliographystyle{unsrt} 
\bibliography{sections/references}

\newpage
\clearpage

\onecolumn

\renewcommand{\thesection}{\Alph{section}}
\renewcommand{\theHsection}{A\arabic{section}}

\beginsupplement

\section*{Appendix}

\section{SmartDCA superiority over DCA over m-buying events}
\label{app:smd}

Notice that Theorems~\ref{thm:sdca} and \ref{thm:usdca} are essentially special cases of Theorem~\ref{thm:bsdca} for $f$ the identity function, and only the proof of Theorem~\ref{thm:bsdca} would therefore be necessary. However, we show here a simpler proof for Theorem~\ref{thm:sdca}.

\sdca*
\begin{proof}
Using the SmartDCA, the quantity $q$ of asset we are going to buy is:

\begin{align}
    q &=\frac{c_b}{p_1}\frac{p_r}{p_1} + \frac{c_b}{p_2}\frac{p_r}{p_2}   + \cdots   + \frac{c_b}{p_m}\frac{p_r}{p_m} \\
    &=c_bp_r\Big(\sum_{i=1}^m\frac{1}{p_i^2}\Big)
\end{align}

On the other hand, the cost of these transactions is:

\begin{align}
    c &= c_b\Big(\frac{p_r}{p_1} + \frac{p_r}{p_2}    + \cdots   + \frac{p_r}{p_m} \Big) \\
    &=c_bp_r\Big(\sum_{i=1}^m\frac{1}{p_i}\Big)
\end{align}

This results in the following average price:

\begin{align}
    \frac{c}{q}
    &= \frac{c_bp_r\sum_{i=1}^m \frac{1}{p_i} }{c_bp_r\sum_{i=1}^m \frac{1}{p_i^2}} \\
    &= \frac{\sum_{i=1}^m \frac{1}{p_i} }{\sum_{i=1}^m \frac{1}{p_i^2}}\label{eq:smart_dca}
\end{align}

Now we need the equivalent quantity in the case where the investor used the standard DCA strategy.
In the DCA case, we have the following:

\begin{align}
    c &= mc_b \\
    q &= c_b\sum_{i=1}^m\frac{1}{p_i} \\
    \frac{c}{q} &= \frac{m}{\sum_{i=1}^m\frac{1}{p_i}}
\end{align}

To establish the superiority of the SmartDCA over the DCA we have to prove the following inequality:

\begin{align}
    \frac{m}{\sum_{i=1}^m\frac{1}{p_i}} &\geq\frac{\sum_{i=1}^m \frac{1}{p_i} }{\sum_{i=1}^m \frac{1}{p_i^2}}\\
    m\sum_{i=1}^m \frac{1}{p_i^2} &\geq\sum_{o=1}^m\frac{1}{p_o}\sum_{i=1}^m \frac{1}{p_i} \\
    & = \Big(\sum_{i=1}^m \frac{1}{p_i}\Big)^2
\end{align}

\noindent where in the second line we rearranged the factors to make the proof easier. Now we start from the left-hand side of the inequality, and we use the Cauchy-Schwarz (CS) inequality to prove that the inequality is actually true:

\begin{align}
    m\sum_{i=1}^m \frac{1}{p_i^2} &= \Big(\sum_{i=1}^m 1^2\Big)\sum_{i=1}^m \frac{1}{p_i^2} \\
    &\geq \Big(\sum_{i=1}^m 1\cdot \frac{1}{p_i}\Big)^2 && \text{CS}\\
    &= \Big(\sum_{i=1}^m \frac{1}{p_i}\Big)^2
\end{align}

\noindent which is exactly what we wanted to prove. QED

\end{proof}

\newpage
\clearpage
\section{Quasi-Lehmer means}
\label{app:newmeans}

Let's define two quasi-Lehmer means:

\begin{align}        
    L^{(out)}_{\rho+1}(\boldsymbol{x})=& \ \frac{\sum_{i=1}^m x_if(x_i)^\rho }{\sum_{i=1}^m f(x_i)^\rho}, \qquad L^{(in)}_{\rho+1}(\boldsymbol{x})=\frac{\sum_{i=1}^m x_if(x_i^\rho) }{\sum_{i=1}^m f(x_i^\rho)}
\end{align}

\noindent as two generalizations of the  Lehmer mean \cite{bullen2003handbook, bullen2013means}:

\begin{align}
    L_\rho(\boldsymbol{x})= \frac{\sum_{i=1}^m x_i^{\rho} }{\sum_{i=1}^m x_i^{\rho-1}}
\end{align}

\noindent that was used to prove Theorem~\ref{thm:usdca}. Now we want to understand if they are monotonic increasing with $\rho$.

\qlm*

\begin{proof}

We proceed by showing that their derivative with respect to $\rho$ is always positive given the Theorem assumptions, to determine that they are monotonic increasing with $\rho$. After taking the derivative and factorizing, in $(\star)$ we split the summation into terms that are $i>j$, $i=j$, and $i<j$, notice that they are zero for  $i=j$, and change the notation from $i,j\rightarrow j,i$, when $i<j$:

\begin{align}
    L^{(out)}_{\rho+1}(\boldsymbol{x})=& \ \frac{\sum_{i=1}^m x_if(x_i)^\rho }{\sum_{i=1}^m f(x_i)^\rho}\\
    \frac{\partial L^{(out)}_{\rho+1}(\boldsymbol{x})}{\partial \rho} =& \ \frac{\splitdfrac{\sum_{j=1}^m f(x_j)^\rho\cdot \sum_{i=1}^mx_if(x_i)^\rho\log f(x_i)}{-\sum_{j=1}^m f(x_j)^\rho\log f(x_j)\cdot \sum_{i=1}^mx_if(x_i)^\rho}}{(\sum_{j=1}^m f(x_j)^\rho)^2} \\
    =& \ \frac{\sum_{i,j} x_if(x_i)^\rho f(x_j)^\rho \Big(\log f(x_i)-\log f(x_j)\Big) }{(\sum_{j=1}^m f(x_j)^\rho)^2}\\
    =& \ \frac{\splitdfrac{\sum_{i>j} x_if(x_i)^\rho f(x_j)^\rho \Big(\log f(x_i)-\log f(x_j)\Big)}{ + x_jf(x_j)^\rho f(x_i)^\rho \Big(\log f(x_j)-\log f(x_i)\Big)}}{(\sum_{j=1}^m f(x_j)^\rho)^2} && \star \\
    =& \ \frac{\sum_{i>j} f(x_i)^\rho f(x_j)^\rho (x_i-x_j)\Big(\log f(x_i)-\log f(x_j)\Big)}{(\sum_{j=1}^m f(x_j)^\rho)^2}\geq 0\label{eq:loutderivative}
\end{align}

Since $f(\cdot)$ and $\log(\cdot)$ are monotonic increasing, then $\log f(\cdot)$ is monotonic increasing. This means that if $x_i>x_j$ then $\log f(x_i)>\log f(x_i)$, by definition of monotonic increasing. Therefore if $x_i - x_j>0$ then $\log f(x_i)-\log f(x_j)>0$ and we have $(x_i - x_j)(\log f(x_i)-\log f(x_j))>0$. In the opposite case, when $x_i<x_j$, by the monotonic increasing property we have $\log f(x_i)<\log f(x_j)$, which we can rewrite as 
$x_i - x_j<0$ implies \mbox{$\log f(x_i)-\log f(x_j)<0$}, and therefore the multiplication of two negative numbers is positive. This proves that all the summands in Eq.~(\ref{eq:loutderivative}) are positive; therefore, the sum is positive.

As a consequence, $L^{(out)}_{\rho+1}(\boldsymbol{x})$ will be monotonic increasing as long as $f$ is positive and monotonic increasing.
However, if we follow the same steps for $L^{(in)}_{\rho+1}(\boldsymbol{x})$, we get:

\begin{align}
    L^{(in)}_{\rho+1}(\boldsymbol{x})=& \ \frac{\sum_{i=1}^m x_if(x_i^\rho) }{\sum_{i=1}^m f(x_i^\rho)} \\
    \frac{\partial L^{(in)}_{\rho+1}(\boldsymbol{x})}{\partial \rho} =& \ \frac{\splitdfrac{\sum_{j=1}^m f(x_j^\rho)\cdot \sum_{i=1}^mx_if'(x_i^\rho)x_i^\rho\log x_i}{-\sum_{j=1}^m f'(x_j^\rho)x_j^\rho\log x_j\sum_{i=1}^mx_if(x_i^\rho)}}{(\sum_{j=1}^m f(x_j^\rho))^2}\\
    =& \ \frac{\sum_{i,j} \Big[f(x_j^\rho)f'(x_i^\rho)x_ix_i^\rho\log x_i-f(x_i^\rho)f'(x_j^\rho)x_ix_j^\rho\log x_j\Big]}{(\sum_{j=1}^m f(x_j^\rho))^2}\\
    =& \ \frac{\sum_{i>j} \Big[\splitdfrac{f(x_j^\rho)f'(x_i^\rho)x_ix_i^\rho\log x_i- f(x_i^\rho)f'(x_j^\rho)x_ix_j^\rho\log x_j}{ +  f(x_i^\rho)f'(x_j^\rho)x_jx_j^\rho\log x_j- f(x_j^\rho)f'(x_i^\rho)x_jx_i^\rho\log x_i}\Big]}{(\sum_{j=1}^m f(x_j^\rho))^2}\\
    =& \ \frac{\sum_{i>j} \Big[\splitdfrac{f(x_j^\rho)f'(x_i^\rho)x_i^\rho\log x_i(x_i-x_j)}{-   f(x_i^\rho)f'(x_j^\rho)x_j^\rho\log x_j(x_i-x_j)}\Big]}{(\sum_{j=1}^m f(x_j^\rho))^2}\\
    =& \ \frac{\sum_{i>j} f(x_j^\rho)f'(x_i^\rho)(x_i^\rho\log x_i-x_j^\rho\log x_j)(x_i-x_j)}{(\sum_{j=1}^m f(x_j^\rho))^2} \\
    =& \ \frac{\sum_{i>j} f(x_j^\rho)\partial f(y)/\partial y \ \rho x_i^{\rho-1}(x_i^\rho\log x_i-x_j^\rho\log x_j)(x_i-x_j)}{(\sum_{j=1}^m f(x_j^\rho))^2} 
\end{align}

\noindent Even when assuming $\rho\geq0$ and $f$ positive and monotonic increasing, it will only be positive if $x^\rho\log x$ is monotonic increasing with $x$. However, we show in the following that it is not generally the case. In fact:
\begin{align}\frac{\partial}{\partial x} x^\rho \log x =& \ \rho x^{\rho-1}\log x + x^{\rho-1} \\
=& \ x^{\rho-1}(\rho \log x + 1)
\end{align}
\noindent which is positive only if $\rho \log x + 1\geq0$ and therefore only for $x \geq \ e^{-\frac{1}{\rho}}$. In other words, $L^{(in)}_{\rho+1}(\boldsymbol{x})$ is  monotonic increasing with $\rho$, if $\rho\geq\max\{-1/\log x,0 \}$ or $\rho\leq\min\{-1/\log x,0 \}$, so it's not in general monotonic increasing with $\rho$. QED.
\end{proof}

For the sake of completeness, we also define the correspondent quasi-Gini means for $\rho+1\neq\gamma$ as:

\begin{align}        
    G^{(out)}_{\rho+1, \gamma}(\boldsymbol{x})=& \ \left( \frac{\sum_{i=1}^m x_if(x_i)^\rho }{\sum_{i=1}^m f(x_i)^\gamma}\right)^{\frac{1}{\rho+1-\gamma}}, \qquad G^{(in)}_{\rho+1, \gamma}(\boldsymbol{x})=\left(\frac{\sum_{i=1}^m x_if(x_i^\rho) }{\sum_{i=1}^m f(x_i^\gamma)}\right)^{\frac{1}{\rho+1-\gamma}}
\end{align}

\noindent that become the quasi-Lehmer means for $\rho=\gamma$.
Note that an analogue to Theorem~\ref{thm:qlm} is also valid for higher quasi-Lehmer moments, defined as:
\begin{align}        
    L^{(out)}_{\rho+1, \xi}(\boldsymbol{x})=& \ \frac{\sum_{i=1}^m x_i^\xi f(x_i)^\rho }{\sum_{i=1}^m f(x_i)^\rho}
\end{align}

\begin{restatable}[quasi-Lehmer moments monotonicity]{theorem}{qlmm}
\label{thm:qlmm} If $\rho\leq \rho'$, $\xi\geq1$ and $f$ is positive and monotonic increasing then $L^{(out)}_{\rho, \xi}(\boldsymbol{x})\leq L^{(out)}_{\rho', \xi}(\boldsymbol{x})$, and therefore $L^{(out)}_{\rho, \xi}(\boldsymbol{x})$ is monotonic increasing with $\rho$. 
\end{restatable}

\begin{proof}

We proceed similarly as to prove Theorem~\ref{thm:qlm}:

\begin{align}
    L^{(out)}_{\rho+1, \xi}(\boldsymbol{x})= & \ \cdots \\
    =& \ \frac{\sum_{i>j} f(x_i)^\rho f(x_j)^\rho (x_i^\xi-x_j^\xi)\Big(\log f(x_i)-\log f(x_j)\Big)}{(\sum_{j=1}^m f(x_j)^\rho)^2}\geq 0
\end{align}

\noindent which is still monotonic increasing for positive monotonic increasing $f$ because $x_i^\xi$ is also monotonic increasing for $\xi\geq1$. QED

\end{proof}

As you can see, the same can be proven for the more general case:

\begin{align}        
    L^{(out)}_{\rho+1, g}(\boldsymbol{x})=& \ \frac{\sum_{i=1}^m g(x_i) f(x_i)^\rho }{\sum_{i=1}^m f(x_i)^\rho}
\end{align}

\begin{restatable}[quasi-Lehmer expectation monotonicity]{theorem}{qlem}
\label{thm:qlem} If $\rho\leq \rho'$,  $f$ is positive monotonic increasing, and $g$ monotonic increasing, then $L^{(out)}_{\rho, g}(\boldsymbol{x})\leq L^{(out)}_{\rho', g}(\boldsymbol{x})$, and therefore $L^{(out)}_{\rho,g}(\boldsymbol{x})$ is monotonic increasing with $\rho$.
\end{restatable}

\noindent where the proof follows the exact same steps as the two previous proofs, but e.g. replacing $x^\xi$ by $g(x)$.

\section{$(f)\rho$-SmartDCA$^{(out)}$ superiority over DCA}
\label{app:outsuperiority}

\bsdca*

\begin{proof}

We proceed as before, we start with the $\rho$-SmartDCA$^{(out)}$, at each time step, we invest an amount proportional to a base cost $c_b$, and take the ratio of the reference price $p_r$ and current price:

\begin{align}
    q  &=c_b\Big(\sum_{i=1}^m\frac{1}{p_i}f\Big(\frac{p_r}{p_i}\Big)^\rho\Big)\\
    c &=c_b\Big(\sum_{i=1}^mf\Big(\frac{p_r}{p_i}\Big)^\rho\Big)
\end{align}

Now consider the quasi-Lehmer \textit{out} mean we defined in the main text:

\begin{align}
    L^{(out)}_{\rho+1}=& \ \frac{\sum_{i=1}^m x_if(x_i)^\rho }{\sum_{i=1}^m f(x_i)^\rho}
\end{align}

We will make use of the result of our Threorem \ref{thm:qlm}, the fact that $\rho\leq\rho'\implies L^{(out)}_\rho(x)\leq L^{(out)}_{\rho'}(x)$.
Using  the notation $r_i = p_r/p_i$,
 if  $\rho\leq\rho'$ we can write

\begin{align}
    \mu_\rho=& \ \frac{c}{q}\\
    =& \ \frac{\sum_{i=1}^m f(\frac{p_r}{p_i})^\rho }{\sum_{i=1}^m \frac{p_r}{p_i}f(\frac{p_r}{p_i})^\rho}\\
    =& \ \frac{1}{L^{(out)}_{\rho+1}(\boldsymbol{r})}\\
    \geq& \ \frac{1}{L^{(out)}_{\rho'+1}(\boldsymbol{r})}=\mu_{\rho'}
\end{align}

\noindent which ends the proof. QED.

\end{proof}

\end{document}